\documentclass{amsart}

\usepackage{url}
\usepackage{graphics}
\usepackage{a4wide}
\usepackage{amsmath, amsfonts}
\usepackage{epstopdf}
\usepackage{graphicx}

\newtheorem{theorem}{Theorem}

\newtheorem{corollary}[theorem]{Corollary}

\newcommand{\T}{{\mathcal T}}
\newcommand{\C}{{\mathcal C}}

\title{A note on convex characters, Fibonacci numbers and exponential-time algorithms}

\author{Steven Kelk}

\author{Georgios Stamoulis}

\address{Department of Data Science and Knowledge Engineering (DKE), Maastricht University, P.O. Box 616, 6200 MD Maastricht, The Netherlands}

\email{steven.kelk@maastrichtuniversity.nl}
\email{georgios.stamoulis@maastrichtuniversity.nl}

\begin{document}

\begin{abstract}
Phylogenetic trees are used to model evolution: leaves are labelled to represent contemporary species (``taxa'') and interior vertices represent extinct ancestors. Informally, convex characters are measurements on the contemporary species in which the subset of species (both contemporary and extinct) that share a given state, form a connected subtree. Given an unrooted, binary phylogenetic tree $\T$ on a set of $n\geq 2$ taxa, a closed
(but fairly opaque) expression for the number of convex characters on $\T$ has been known since 1992, and this is independent of the exact topology of $\T$. In this note we prove that
this number is actually equal to the $(2n-1)$th Fibonacci number. Next, we define $g_k(\T)$ to be the number of convex characters on $\T$ in which each state appears on at least $k$ taxa. We show that, somewhat curiously, $g_2(\T)$ is also independent of the topology of $\T$, and is equal to
to the $(n-1)$th Fibonacci number. As we demonstrate, this topological neutrality subsequently breaks down for $k \geq 3$. However, we show that for each fixed $k \geq 1$, $g_k(\T)$ can be computed in $O(n)$ time
and the set of characters thus counted can be efficiently listed and sampled. We use these
insights to give a simple but effective exact algorithm for the NP-hard \emph{maximum parsimony distance}
problem that runs in time $\Theta( \phi^{n} \cdot n^2 )$, where $\phi \approx 1.618...$ is the golden ratio, and an exact algorithm which computes the \emph{tree bisection and reconnection} distance (equivalently, a \emph{maximum agreement forest}) in time $\Theta( \phi^{2n}\cdot \text{poly}(n))$, where $\phi^2 \approx 2.619$.
\end{abstract}

\maketitle

\section{Introduction}
\label{sec:intro}

\noindent
Phylogenetics is the science of accurately and efficiently inferring evolutionary trees given only information about contemporary species \cite{semple_steel_2003}. An important concept within phylogenetics is \emph{convexity}. Essentially this captures the situation when, within a phylogenetic (i.e. evolutionary) tree, each biological state emerges exactly once: it should not emerge, die out, and then re-emerge. More concretely, given a phylogenetic tree and a set of states assigned to its leaves, can we assign states to the internal vertices of the tree such that each state forms a connected ``island'' within the tree? If this is possible, the assignment of states to the leaves is known as a \emph{convex character}. 

In this article we present a number of results concerning the enumeration of convex characters. In Section \ref{sec:prelim} we give formal definitions and describe relevant earlier work. In Section \ref{sec:result} we start by showing that an earlier result counting convex characters can be simplified to a term of the Fibonacci sequence. We then seek to count convex characters with the added restriction that each state should occur on at least $k$ leaves, proving the somewhat surprising result that (as for $k=1$) tree topology is irrelevant for $k=2$, and that a formulation in terms of Fibonacci numbers is again possible. We give an explicit example showing that for $k \geq 3$ the topological neutrality breaks down. In Section \ref{sec:dp} we show that for all $k$ the size of the space can be counted in polynomial time and space using dynamic programming, which also permits listing and sampling uniformly at random, noting also that non-isomorphic trees can have exactly the same vector of space sizes (for $k=1, 2, ...$). In Section \ref{sec:alg} we give a number of algorithmic applications for NP-hard problems arising in phylogenetics that seek to quantify the dissimilarity of two phylogenetic trees. Finally, in Section \ref{sec:conc} we briefy discuss a number of open problems arising from this work. The software associated with this article has been made publicly available.

\section{Preliminaries}\label{sec:prelim}

 For general background on mathematical phylogenetics we refer to \cite{semple_steel_2003,dress2012basic}. An {\it unrooted binary phylogenetic $X$-tree} is an undirected tree $\T =(V(\T),E(\T))$ where every internal vertex has degree 3 and whose leaves are bijectively labelled by a set $X$, where $X$ is
often called the set of \emph{taxa} (representing the contemporary species). We use $n$ to denote $|X|$ and often simply write \emph{tree} when this is clear from the context.

A {\it character} $f$ on $X$ is a surjective function $f: X\rightarrow \C$ for some set $\C$ of {\em states} (where a state represents some characteristic of the species e.g. number of legs). We say that $f$ is an $r$-state character if $|\C|=r$.
Each character naturally induces a partition of $X$ and here we regard two characters as being equivalent if they both induce the same partition of $X$. An {\it extension} of a character $f$ to $V(\T)$ is a function $h: V(\T)\rightarrow \C$ such that $h(x) = f(x)$ for all $x$ in $X$. For such an extension $h$ of $f$, we denote by $l_{h}(\T)$ the number of edges $e=\{u,v\}$ such that $h(u) \neq h(v)$. The {\em parsimony score} of a character $f$ on $\T$, denoted by $l_{f}(\T)$, is obtained by minimizing $l_{h}(\T)$ over all possible extensions $h$ of $f$. We say that a character $f: X \rightarrow \C$ is \emph{convex on $\T$} if $l_{f}(\T) = |\C| - 1$. Equivalently: a character $f: X \rightarrow \C$ is convex on $\T$ if there exists an extension $h$ of $f$ such that, for each state $c \in \C$, the vertices of $\T$ that are allocated state $c$ (by $h$) form a connected subtree of $\T$. We call such an extension $h$ a \emph{convex extension} of $f$.  See Figure
\ref{fig:g2} for an example. The convexity of a character can be tested in polynomial \cite{fitch_1971,hartigan1973minimum} (in fact, linear \cite{bachoore2006convex}) time.

\begin{figure}[h]
\centering
\includegraphics[scale=0.2]{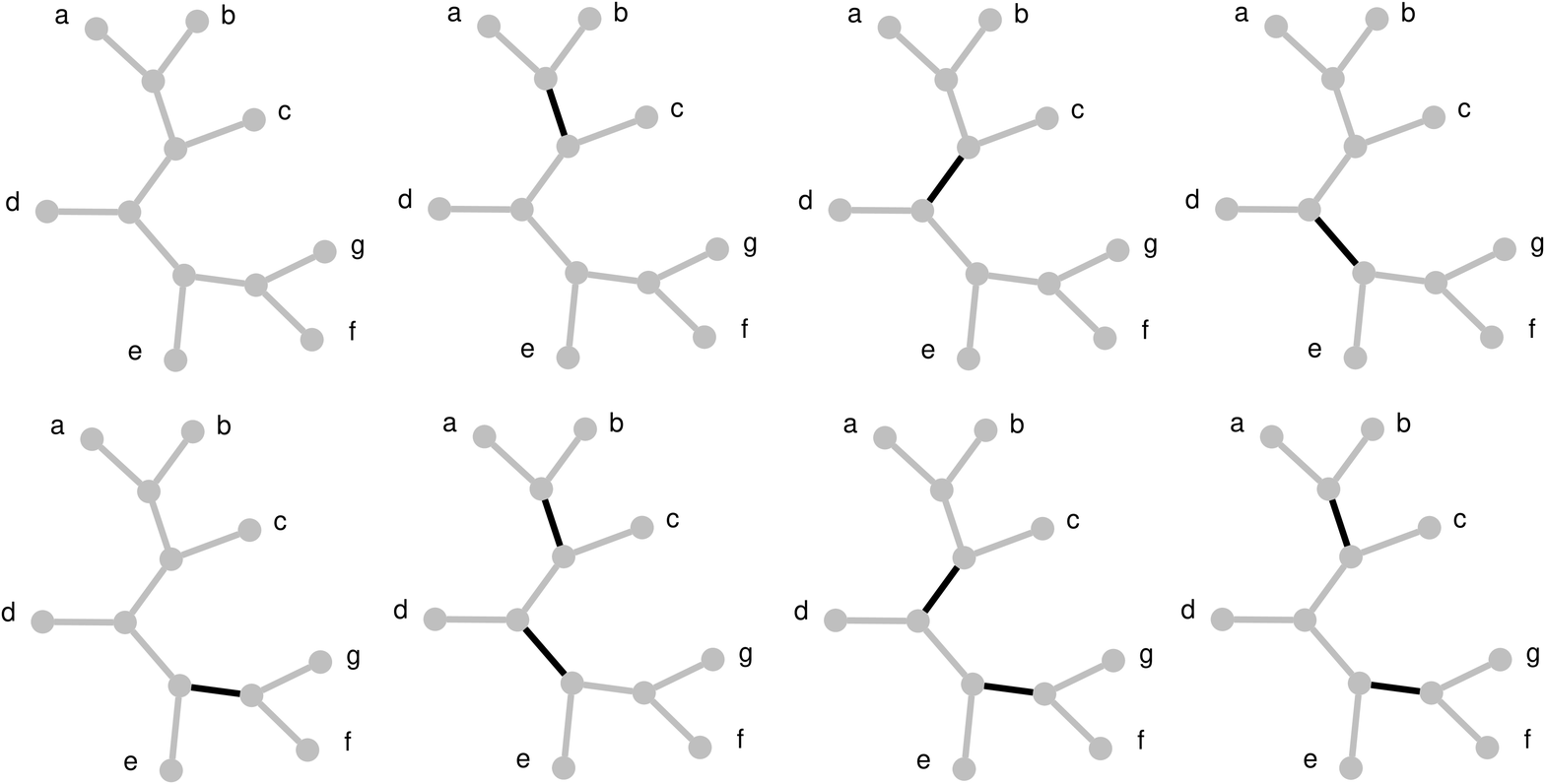}
\caption{For the given tree $\T$ (on 7 taxa) there are 233 convex characters in total,
but only 8 in which each state appears on at least 2 taxa, and these are shown above. 1
character uses exactly 1 state ($abcdefg$), so $g_2(\T,1)=1$, 4 characters use 2 states ($ab|cdefg$, $abc|defg$, $abcd|efg$ and $abcde|fg$), so $g_2(\T,2)=4$ and 3 characters use 3 states ($ab|cd|efg$, $abc|de|fg$ and $ab|cde|fg$), so $g_2(\T,3)=3$. For each character we have shown an extension verifying that the subtree induced by each state is connected i.e. that the character is convex.}
\label{fig:g2} 
\end{figure}

We write $g(\T,r)$ to denote the number of $r$-state convex characters on $\T$ and 
$g_i(\T, r)$ $(i \geq 1)$ to denote the number of those characters that have the additional
property that each state used by the character appears on at least $i$ taxa. It follows from the definition of character that $g_1(\T,r) = g(\T,r)$. We define:
\[
g_i(\T) = \sum_{r=1}^{n} g_i(\T,r).
\]
The value $g_1(\T)$ is therefore equal to the total number of convex characters on $\T$. For
the tree shown in Figure \ref{fig:g2}, $g(\T)=g_1(\T)=233$ and $g_2(\T)=8$. 
We adopt the standard convention that the binomial coefficient $\binom{n}{k}$ evaluates to $1$ if $k=0$,  and $0$ if $n < k$ and $k > 0$. In \cite{steel1992complexity} it is proven that, for $n, r \geq 1$,
\[
g(\T,r) = g_1(\T,r) = \binom{2n - r -1}{r - 1}.
\]
Hence,
\begin{equation}
\label{eq:g1}
g_1(\T) = \sum_{r=1}^{n} \binom{2n - r -1}{r - 1}.
\end{equation}
As observed in \cite{steel1992complexity}, the expression for $g_1(\T,r)$ (somewhat surprisingly) does not depend on the topology of $\T$, only on the number of taxa $n$. Hence we can write $g(n), g_1(n)$ and
$g_1(n,r)$ without ambiguity.
\section{Fibonacci numbers and convex characters}\label{sec:result}

\begin{theorem}
\label{thm:poom}
The value $g_2(\T,r)$ only depends on $n$ (i.e. the topology of $\T$ is not relevant) and
for $n \geq 2, r \geq 1$ is given by the expression
\[
g_2(\T,r) = \binom{n-r-1}{r-1}.
\]
\end{theorem}
\begin{proof}
We prove this by induction on $n$. For the base case note that for $n \in \{2,3,4\}$ there is only
one binary tree topology (up to relabelling of taxa) possible on $n$ taxa and that the expression correctly evaluates to 1 when $r=1$ and, when $r \geq 2$, evaluates to 0 in all cases except
$n=4, r=2$ when it correctly evaluates to 1.

Consider then $n \geq 5$, and let $r$ be any value $2 \leq r \leq \lfloor n/2 \rfloor$. (For $r=1$
the expression correctly evaluates to 1, and for $r > \lfloor n/2 \rfloor$ the expression
correctly evaluates to 0). Every tree on $4$ or more taxa contains at least one \emph{cherry}: two taxa $x,y$ that have a common parent $u$ where the third neighbour of $u$ is an interior
vertex. Fix such a cherry. (A similar
technique is used in \cite{steel1992complexity} and \cite{steel1995classifying}).
Observe that any convex character $f$ on $\T$ with the property that each state appears at least twice,  has $f(x)=f(y)$. This follows from the connected-subtree definition of convexity.
Now, let $\T' = \T| (X \setminus \{ x\})$ and let $\T'' = \T | (X \setminus \{x,y\})$, where
$\T|X'$  denotes the tree (on the set of taxa $X'$) obtained from $\T$ by taking the minimum subtree connecting the elements of $X'$ and then suppressing vertices of degree 2.

 There are two cases to distinguish. The first case is when the state $f(x)=f(y)$ does not appear on any other taxa. There are $g_2(\T'', r-1)$ such characters. The second case is when $f(x)=f(y)$ does appear on at least one other taxon. There
are $g_2(\T', r)$ such characters. Hence,
\begin{equation}
\label{eq:recur}
g_2(\T,r) = g_2(\T'',r-1) + g_2(\T',r).
\end{equation}
By the inductive hypothesis we have
\begin{align*}
g_2(\T,r) &= \binom{(n-2)-(r-1)-1}{(r-1)-1} + \binom{(n-1)-r-1}{r-1} \\
&= \binom{n-r-2}{r-2} + \binom{n-r-2}{r-1}\\
&= \binom{n-r-1}{r-1}.
\end{align*}
The last equality follows from the well-known identity known as \emph{Pascal's Rule} i.e.
$\binom{n}{k} = \binom{n-1}{k} + \binom{n-1}{k-1}$, which holds for $1 \leq k \leq n$. This completes the proof.
\end{proof}
Consequently, the total number of convex characters on a tree $\T$ with each state appearing at least twice is independent of the topology of $\T$. Specifically, we have:
\begin{equation}
\label{eq:g2}
g_2(\T) = g_2(n) = \sum_{r=1}^{n} \binom{n-r-1}{r-1} =  \sum_{r=1}^{\lfloor n/2 \rfloor} \binom{n-r-1}{r-1}
\end{equation}
\begin{corollary}
\label{cor:switch}
For even $n$, $g_2(n) = g_1( n/2 )$.
\end{corollary}
\begin{proof}
This is immediate by observing that Equation (\ref{eq:g2}) can be obtained by substituting $n/2$ for $n$ in Equation (\ref{eq:g1}). 
\end{proof}

Let $F(n)$ denote the $n$th Fibonacci number. That is, $F(0)=0$, $F(1)=1$ and for $n \geq 2$,
$F(n) = F(n-1) + F(n-2)$. For comprehensive background on Fibonacci numbers see \cite{koshy2011fibonacci}.
\begin{theorem}
\label{thm:fiblink}
For $n \geq 2$, $g_2(n) = F(n-1)$ and $g_1(n) = F(2n-1)$.
\end{theorem}
\begin{proof}
The following identity is classical\footnote{This is usually attributed to \'{E}douard Lucas (1876). By applying Pascal's Rule and some algebraic manipulation it can be proven by induction.} $(n \geq 0)$:
\[
F(n+1) = \sum_{k=0}^{\lfloor n/2 \rfloor} \binom{n-k}{k}.
\]
If we index $k$ from $1$ rather than 0 we obtain,
\[
F(n+1) = \sum_{k=1}^{\lfloor n/2 \rfloor + 1 } \binom{n-(k-1)}{k-1}.
\]
Now, if we replace $n$ with $n-2$:
\begin{align*}
F(n-1) &= \sum_{k=1}^{\lfloor (n-2)/2 \rfloor + 1 } \binom{n-k-1}{k-1} \\
&= \sum_{k=1}^{\lfloor n/2 \rfloor  } \binom{n-k-1}{k-1}\\
&= g_2(n).
\end{align*}
The expression for $g_1(n)$ is then obtained by applying Corollary \ref{cor:switch}. 
\end{proof}

The question arises whether the values $g_3(\T,r)$ and/or $g_3(\T)$ share the topological
neutrality of their $g_1$ and $g_2$ counterparts. This is not the case: see Figure \ref{fig:g3}. Here $g_3(\T_1)=5$, because $g_3(\T_1,1)=1$, $g_3(\T_1,2)=3$, $g_3(\T_1,3)=1$ and $g_3(\T_1,r)=0 $ (for $r > 3$). However, $g_3(\T_2)=6$, because $g_3(\T_2,1)=1$, $g_3(\T_2,2)=4$, $g_3(\T_2,3)=1$ and $g_3(\T_2,r)=0$ (for $r >3$). 

\begin{figure}[h]
\centering
\includegraphics[scale=0.2]{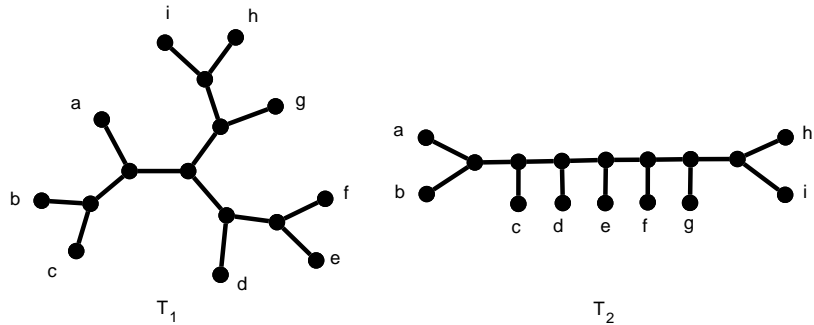}
\caption{The number of characters that are convex on $\T_1$ with each state appearing
on at least 3 taxa, is not the same as the corresponding number for $\T_2$ i.e.
$g_3(\T_1) \neq g_3(\T_2)$. This is because $g_3(\T_1,2)=3$ (the characters $abc|defghi, abcdef|ghi$ and $abcghi|def$) while $g_3(\T_2,2)=4$ (the characters
$abc|defghi, abcd|efghi, abcde|fgh$ and $abcdef|ghi$). Hence
topology does play a role here, contrasting with the situation for $g_1$ and $g_2$.}
\label{fig:g3} 
\end{figure}

\section{Computing $g_k(\T)$ (and listing its elements) with dynamic programming}
\label{sec:dp}

The results from the previous section give rise to a number of questions. Can we compute $g_k(\T)$ in polynomial time, for $k \geq 3$? Also, if we want to explicitly list all the elements counted by $g_k(\T)$ ($k \geq 1)$, is it possible to achieve this in some reasonable total running time e.g. $O( g_k(\T) \cdot \text{poly}(n) )$? We show that, for all $k \geq 1$, the answer to both questions is \emph{yes}. Specifically, we show how to compute $g_k(\T)$ using dynamic programming, for each $k \geq 1$. The combinatorial recurrence within the dynamic programming will also allow us to derive a polynomial-time computable bijection from $\{1, 2, \ldots, g_k(\T) \}$ to the characters counted by $g_k(\T)$. Using this bijection it is then straightforward to list (or sample) these elements. Note that this is also an advance for $g_1(\T)$, since the recurrence used in \cite{steel1992complexity} to derive $g_1(n)$ is based on inclusion-exclusion: it is not obvious how to transform it into a bijection.

 We begin by rooting $\T$ by subdividing an arbitrary edge with a new vertex and (implicitly) directing all edges away from this new vertex. The new vertex becomes the root of the tree. Note that this rooting operation has no impact on the convexity of characters, and the location of the root is irrelevant; it is simply a convenience which ensures that the term ``child'' is well-defined. The dynamic programming works bottom-up, from the leaves towards the root.


Here it is helpful to represent a character $f$ on $X$ as a set of non-overlapping, non-empty subsets $\{B_1, \ldots, B_t\}$ that partition $X$, where each $B_i$ corresponds to a state.
We also need some new definitions. A
character $f$ is \emph{valid} for $g_k(\T)$ if
\begin{itemize}
\item[--] $f$ is convex on $\T$, and
\item[--] for each $B_i \in f$, $|B_i| \geq k$.
\end{itemize} 
Consider an ordered pair $(f, B)$ where $f$ is a character on $X$ and $B \in f$. We call such a pair a \emph{character-root} pair of $\T$ if
\begin{itemize}
\item[--] $f$ is convex on $\T$, and
\item[--] there exists a convex extension of $f$ in which the root of $\T$ is assigned state $B$.
\end{itemize} 
Equality between character-root pairs is defined strictly i.e. $(f,B) = (f',B')$ if and only if $f=f'$ and $B=B'$. We say that a character-root pair $(f,B)$ of $\T$ is \emph{semi-valid} for $g_k(\T)$ if, for each $B_i \in f$, $B_i \neq B \Rightarrow |B_i| \geq k$. Note that if a character-root pair $(f,B)$ is semi-valid for $g_k(\T)$, then $f$ is valid for $g_k(\T)$ if and only if $|B| \geq k$.

At each vertex $u$ of the tree we will compute and store the following $k+1$ values, where $\T_u$ is simply the subtree rooted at $u$:
\begin{itemize}
\item[--] $g_k(\T_u)$,
\item[--] for each $1 \leq m \leq k-1$ the number $h(\T_u, m)$ which is defined as the number of
character-root state pairs $(f,B)$ of $\T_u$ such that both the following conditions hold: $(f,B)$ is semi-valid for $g_k(\T_u)$ and  $|B| = m$. We also store $h(\T_u, k)$ which
is defined slightly differently: we replace the term $|B| = m$ with $|B| \geq k$.
\end{itemize}
If $u$ is a taxon, then:
\begin{itemize}
\item[--] $g_k(\T_u)$ is equal to 1 if $k=1$, and $0$ if $k>1$,
\item[--] $h(\T_u, m)$ is equal to 1 if $m=1$ and $0$ if $m > 1$.
\end{itemize}

We show how to compute these values recursively, assuming the corresponding values have already been computed for $\T_l$ (the subtree rooted at the left child of $u$) and $\T_r$ (the subtree rooted at the right child of $u$). First,
\begin{equation}
\label{eq:g}
g_k(\T_u) = g_k(\T_l) g_k(\T_r) + \sum_{\substack{1 \leq i, j \leq k \\ i+j \geq k}} h(\T_l, i)h(\T_r, j).
\end{equation}
The idea behind this recurrence is that characters counted by $g_k(\T_u)$ can be created in two ways: (i) by taking the union of a character from the left subtree with a character from the right subtree, and (ii) by taking a character-root pair $(f, B)$ from the left subtree, a character-root pair $(f', B')$ from the right subtree, and then merging the root states to yield a character
$(f \setminus B) \cup (f' \setminus B') \cup (B \cup B')$. Characters from the subtrees can only be used for (i) if they are already valid (with respect to $g_k$) in their subtrees. The characters $f$ and $f'$ that are used for (ii) might not be valid with respect to their subtrees, but we do require that they can be combined to obtain a character that \emph{is} valid for $g_k(\T_u)$.  This is possible if and only if $(f,B)$ and $(f',B')$ are semi-valid for their respective subtrees and the sum of the cardinalities of $B$ and $B'$ is at least $k$.  See Figure \ref{fig:rootMerge} for an example.

\begin{figure}[h]
\centering
\includegraphics[scale=0.2]{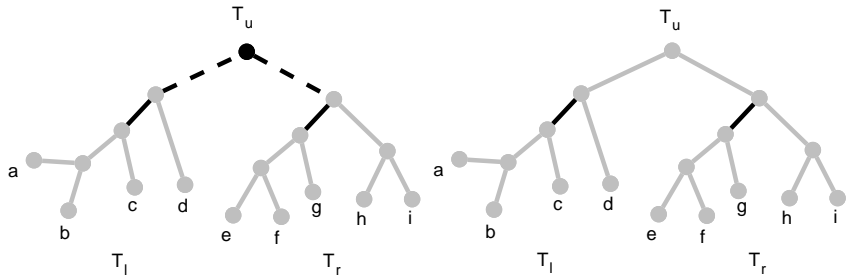}
\caption{The character $abc|d$ is not valid for $g_3(\T_l)$, and $efg|hi$ is not valid for $g_3(\T_r)$, but the character $abc|dhi|efg$ is valid for $g_3(\T_u)$. The validity is obtained
by allowing $d$ and $hi$ to merge, which was possible because they could both ``reach'' the roots of their respective subtrees.}
\label{fig:rootMerge} 
\end{figure}

Second, for $1 \leq m \leq k-1$ we have,
\begin{equation}
\label{eq:h}
h(\T_u, m) =  g_k(\T_l)h(\T_r,m) + h(\T_l,m)g_k(\T_r) + \sum_{\substack{1 \leq i, j \leq k-2 \\ i+j = m}} h(\T_l, i)h(\T_r, j).
\end{equation}
Note that here
$m \leq k-1$, which means that the semi-valid character-root pairs $(f,B)$ counted by this recurrence are such that $f$ is \emph{not} valid
for $g_k(\T_u)$. The first two terms of the recurrence concern the situation analogous to (i) above. Specifically, in this case we assume that no states are merged, so a new semi-valid character-root pair can be created for $\T_u$ if and only if it is constructed from the combination of a valid character from one subtree, with a semi-valid character-root pair from the other. The summation term corresponds to (ii). That is,  we only count combinations of character-root pairs from the two subtrees if the cardinality of the merged state is exactly $m$. 

Finally we have,
\begin{equation}
\label{eq:h2}
h(\T_u, k) =  g_k(\T_l)h(\T_r,k) + h(\T_l,k)g_k(\T_r) + \sum_{\substack{1 \leq i, j \leq k \\ i+j \geq k}} h(\T_l, i)h(\T_r, j).
\end{equation}

This final recurrence is semantically very similar to the previous one. The main difference is that it counts all semi-valid character-root pairs $(f,B)$ for $g_k(\T_u)$ such that $f$ \emph{is} valid for $g_k(\T_u)$.\\
\\
For a given vertex $u$, Equation (\ref{eq:g}) can be computed in $O(k^2)$ time, assuming the values for $\T_l$ and $\T_r$ have already been computed earlier. The same time bound holds for Equations (\ref{eq:h}) (for a specific $1 \leq m \leq k-1$) and  (\ref{eq:h2}). Equation
(\ref{eq:h}) has to be computed for
each $m$, yielding a na\"ive running time bound of $O(k^3)$ (per vertex $u$), but this can easily be improved to $O(k^2)$ by observing that a single $1 \leq i, j \leq k$ sweep over the $h(\T_l, i)$ and $h(\T_r,j)$ values can be recycled for computation of all the different $h(\T_u,m)$ values. There are $2(n-1)$ vertices in the tree. This yields the following theorem.

\begin{theorem}
\label{thm:countdp}
Let $\T$ be an unrooted binary tree on $n$ taxa. For each $1 \leq k \leq n$, $g_k(\T)$ can be computed in $O(k^2 \cdot n)$ time and $O(k \cdot n)$ space.
\end{theorem}

\begin{corollary}
\label{cor:listdp}
Let $\T$ be an unrooted binary tree on $n$ taxa. For each $1 \leq k \leq n$, all the characters that are counted by $g_k(\T)$ can be generated in $O( g_k(\T) \cdot k^2 \cdot n )$ total time, and a character counted by $g_k(\T)$ can be sampled uniformly at random in $O(k^2 \cdot n)$ time and $O(k \cdot n)$ space.
\end{corollary}
\begin{proof}
Critically, there is no inclusion-exclusion involved Equations (\ref{eq:g}), (\ref{eq:h}) and (\ref{eq:h2}). This allows us to impose a canonical ordering on the characters (and character-root pairs) counted by these equations. For example, within Equation (\ref{eq:g}) we can choose to place the type-(i) characters earlier in the ordering than the type-(ii) characters. Within the $g_k(\T_l) g_k(\T_r)$ type-(i)
characters we can refine the order as follows: the \emph{first} character from the left subtree combined with in turn each of the $g_k(\T_r)$ characters from the right subtree, then the \emph{second} character from the
  left subtree combined in turn with each of the $g_k(\T_r)$ characters from the right subtree, and so on. Once a canonical ordering has been chosen and the dynamic programming has been completed, we can start at the root of $\T$ and (using the $g$ and $h$ values computed at all vertices of the tree) recursively backtrack to generate the uniqely defined $i$th character. Hence, we obtain
a bijection from $\{1, 2, \ldots, g_k(\T) \}$ to the characters counted by $g_k(\T)$. The time and space requirements for backtracking through the tree (i.e. evaluating the bijection for a given element of $\{1, 2, \ldots, g_k(\T) \}$) are dominated by the time and space requirements of executing the original dynamic program, which are $O(k^2 \cdot n)$ and $O(k \cdot n)$ respectively. This bijection
can then be used to list all the characters counted by  $g_k(\T)$ or to sample uniformly at random from this space.
\end{proof}

We have implemented the dynamic programming (and the corresponding algorithms for listing and sampling) in Java and this can be downloaded from \url{http://skelk.sdf-eu.org/convexcount}.\\
\\
Finally within the section, for an unrooted binary tree $\T$ on $n$ leaves, we define the $g$-\emph{spectrum} as simply the vector
$(g_1(\T), g_2(\T), \ldots, g_n(\T))$. It is natural to ask whether two trees on $n$ taxa have the same $g$-spectrum if and only if they are isomorphic (see e.g. \cite {bordewich2005identifying} for related discussions of ``identifiability''). Using the code above we have verified that, while this claim is true for $n \leq 10$ leaves (see the software website for a proof), a counter-example exists for $n=11$, see Figure \ref{fig:spectrum}. 
\begin{figure}[h]
\centering
\includegraphics[scale=0.2]{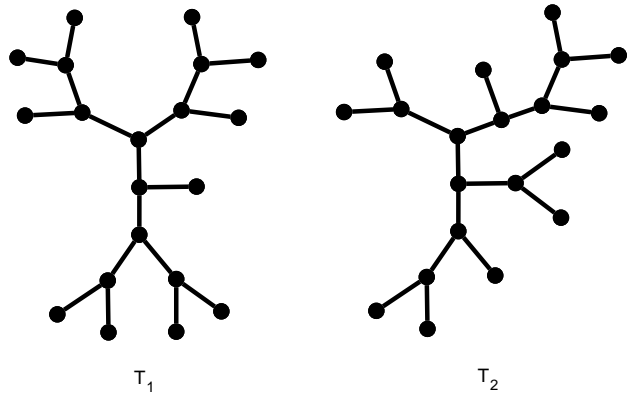}
\caption{These two trees on 11 leaves are non-isomorphic but have the same $g$-spectrum: $(10946, 55, 8, 3, 2, 1, 1, 1, 1, 1, 1)$.}
\label{fig:spectrum} 
\end{figure}

\section{Algorithmic applications}\label{sec:alg}

One of the advantages of expressing $g_1(n)$ and $g_2(n)$ as Fibonacci numbers is that it allows us to give tight bounds on their rate of growth. This can be particularly useful when bounding the running time of algorithms. Consider the following classical, closed-form expression for the Fibonacci numbers $(n \geq 0)$,
where $\phi = \frac{1 + \sqrt{5}}{2} \approx 1.618...$ is the \emph{golden ratio}:
\[
F(n) = \bigg \lfloor \frac{\phi^n}{\sqrt{5}} + \frac{1}{2} \bigg \rfloor.
\]
(It is obtained from Binet's Formula $F(n)= \frac{\phi^{n} - (-\phi)^{-n}}{\sqrt{5}}$ (1843)
by observing that the term $(-\phi)^{-n}$ is vanishing.) Combining with Theorem \ref{thm:fiblink} we obtain
\begin{align*}
g_1(n) &= \bigg \lfloor \frac{\phi^{2n-1}}{\sqrt{5}} + \frac{1}{2} \bigg \rfloor,\\
g_2(n) &= \bigg \lfloor \frac{\phi^{n-1}}{\sqrt{5}} + \frac{1}{2} \bigg \rfloor.\\
\end{align*}
Using asymptotic notation, it is clear that there are $\Theta( \phi^{2n} )$ convex
characters and $\Theta( \phi^{n} )$ convex characters in which each state occurs on at least two taxa. We give two examples of how these insights yield non-trivial exponential-time algorithms for two NP-hard problems arising in phylogenetics.

\subsection{Computation of maximum parsimony distance.}
Let $\T_1, \T_2$ be two unrooted binary trees on the same set of taxa $X$. The metric
$d_{MP}(\T_1, \T_2)$ (the \emph{maximum parsimony distance} of $\T_1$ and $\T_2$) is defined as follows, where $f$ ranges over \emph{all} characters on $X$ and $l_f(.)$ is as defined in Section \ref{sec:prelim}:
\[
d_{MP}(\T_1,\T_2) = \max_{f} | l_f(\T_1) - l_f(\T_2) |
\]
It is NP-hard to compute $d_{MP}$ \cite{kelk2014complexity,fischer2014maximum}. It can be used to quantify the dissimilarity
of two phylogenetic trees and is a lower bound on the similarly NP-hard \emph{tree bisection and reconnection} (TBR) distance, denoted $d_{TBR}$ \cite{AllenSteel2001}. 
\begin{theorem}
Given two unrooted binary trees $\T_1, \T_2$ on the same set of taxa $X$, where $|X|=n$,
$d_{MP}(\T_1, \T_2)$ can be computed in time $\Theta( \phi^{n} \cdot n^2)$,
where $\phi \approx 1.618...$ is the golden ratio. 
\end{theorem}
\begin{proof}
In \cite{kelk2014complexity,fischer2014maximum} it is proven that the optimum is achieved
by some character $f$ that is convex on $\T_1$ or $\T_2$ and where
each state in the character occurs on at least two taxa. Hence simply looping through all the
characters counted by $g_2(\T_1)$ and, separately, all the characters counted by $g_2(\T_2)$ is sufficient to locate an optimal character. Note that $l_f(.)$ can be computed in $O(n^2)$ time using Fitch's algorithm\footnote{Fitch's algorithm has running time $O(ns)$ where $n$ is the number of taxa and $s$ is the
number of states in the character. In our context $s$ can rise to $O(n)$.} \cite{fitch_1971} or dynamic programming. Hence, scoring each character $f$ can easily be performed in quadratic time. The result then follows by leveraging
Corollary \ref{cor:listdp}.
\end{proof}

We have implemented the $d_{MP}$ algorithm in Java and for an exponential-time algorithm the results are encouraging; the code is freely available at \url{http://skelk.sdf-eu.org/convexmpdist}.  On a single 32-bit 1.66GHz Intel Atom (N450) processor the algorithm terminates for $n=20,25,30$ in less than 1 second, 3 seconds and 51 seconds respectively. On a more powerful 64-bit 3.10GHz machine the previously fastest algorithm, the Integer Linear Programming (ILP) approach described in \cite{kelk2014complexity}, took 70 seconds to terminate on 12 taxa, and stalled completely on trees with more than 16 taxa, even using state-of-the-art ILP software. The enhanced range of our software has been recently used in experiments to verify that $d_{MP}$ is often a very good lower bound on $d_{MP}$ \cite{kelk2015reduction}.

\subsection{Computation of TBR distance and maximum agreement forests.}
Finally, we note that the results in this article also give an easy (although, in some cases, somewhat crude) upper bound on the number of
\emph{agreement forests} of two unrooted binary trees $\T_1, \T_2$ on $n$ taxa. Recall that, for an unrooted binary phylogenetic tree $\T$ on
$X$ and $X' \subseteq X$, $\T|X'$ is defined to be the unrooted binary phylogenetic tree on $X'$ obtained by taking the minimal subtree of $\T$
that connects $X'$, and suppressing vertices of degree 2. An agreement forest is a partition of $X$ into non-empty subsets
$X_1, \ldots, X_k$ such that (i) within $\T_1$ (respectively, $\T_2$) the minimal connecting subtrees induced by the $X_i$ are vertex-disjoint and (ii) for
each $X_i$, $\T_1|X_i = \T_2|X_i$ (where here equality explicitly takes the taxa into account). See
\cite{AllenSteel2001} and recent articles such as \cite{chen2015parameterized} for further background on agreement forests. A \emph{maximum agremeent forest} is an agreement forest with a minimum
number of components, and this minimum is denoted $d_{MAF}$. Note that due to part (i) of the definition
every agreement forest induces a character that is convex on both $\T_1$ and
$\T_2$ (although not all characters that are convex on both $\T_1$ and $\T_2$ necessarily correspond to
agreement forests\footnote{However, if $\T_1 = \T_2$ then convex characters and agreement forests are related one-to-one.}). Hence there are at most $O( \phi^{2n} )$ agreement forests, which is
$O( 2.619^{n} )$. It is well-known that $d_{TBR}$ is equal to the number of components in a maximum agreement forest, minus 1 \cite{AllenSteel2001}. Hence, again by leveraging Corollary \ref{cor:listdp} we obtain:
\begin{theorem}
Given two unrooted binary trees $\T_1, \T_2$ on the same set of taxa $X$, where $|X|=n$,
$d_{TBR}(\T_1, \T_2) = d_{MAF}(\T_1, \T_2) - 1$ can be computed in time $\Theta( 2.619^n \cdot \emph{poly}(n))$. Moreover, all agreement forests of $\T_1$ and $\T_2$ can be listed in the same time bound.
\end{theorem}

\section{Conclusion}
\label{sec:conc}
A number of interesting open problems remain. For example, can we characterize non-isomorphic trees that have the same $g$-spectrum? For a given $k \geq 3$ and $n$, can we give analytical lower and upper bounds on $g_k(\T)$, ranging over the space of all trees $\T$ on $n$ taxa?

\section{Acknowledgements}

We thank Mike Steel for helpful discussions.

\bibliographystyle{plain}
\bibliography{Fibonacci2016}

\end{document}